\documentclass[12pt]{article}

\usepackage{amsmath}
\usepackage{amssymb}
\usepackage{amsfonts}
\usepackage{latexsym}
\usepackage{color}

\catcode `\@=11 \@addtoreset{equation}{section}
\def\theequation{\arabic{section}.\arabic{equation}}
\catcode `\@=12

\newcommand{\be}{\begin{equation}}
\newcommand{\en}{\end{equation}}
\newcommand{\bea}{\begin{eqnarray}}
\newcommand{\ena}{\end{eqnarray}}
\newcommand{\beano}{\begin{eqnarray*}}
\newcommand{\enano}{\end{eqnarray*}}
\newcommand{\bee}{\begin{enumerate}}
\newcommand{\ene}{\end{enumerate}}

\newcommand{\mc}{\mathcal}

\newcommand{\D}{{\mc D}}

\newcommand{\LL}{\mc L}

\newcommand{\Sc}{{\cal S}}

\newcommand{\F}{{\cal F}}

\newcommand{\Lc}{{\cal L}}

\newcommand{\1}{1 \!\! 1}

\newcommand{\LD}{{\LL}^\dagger (\D)}

\newcommand{\Hil}{\mc H}

\newtheorem{thm}{Theorem}

\newtheorem{defn}[thm]{Definition}

\newenvironment{proof}{\noindent {\bf Proof --}}{\hfill$\square$ \vspace{3mm}\endtrivlist}

\catcode `\@=11 \@addtoreset{equation}{section}
\catcode `\@=12

\begin{document}

\thispagestyle{empty}

\vspace*{2cm}

\begin{center}
{\Large \bf Abstract ladder operators and their applications}   \vspace{2cm}\\

{\large F. Bagarello}\\
  Dipartimento di Ingegneria,
Universit\`a di Palermo,\\ I-90128  Palermo, Italy\\
and I.N.F.N., Sezione di Napoli\\
e-mail: fabio.bagarello@unipa.it\\
home page: www1.unipa.it/fabio.bagarello

\end{center}

\vspace*{2cm}

\begin{abstract}
\noindent We consider a rather general version of ladder operator $Z$ used by some authors in few recent papers, $[H_0,Z]=\lambda Z$ for some $\lambda\in\mathbb{R}$, $H_0=H_0^\dagger$, and we show that several interesting results can be deduced from this formula. Then we extend it in two ways: first we replace the original equality with formula $[H_0,Z]=\lambda Z[Z^\dagger, Z]$, and secondly we consider $[H,Z]=\lambda Z$ for some $\lambda\in\mathbb{C}$, $H\neq H^\dagger$. In both cases many applications are discussed. In particular we consider factorizable Hamiltonians and Hamiltonians written in terms of operators satisfying the generalized Heisenberg algebra or the $\D$ pseudo-bosonic commutation relations. 
\end{abstract}

\vspace{2cm}


\vfill


\newpage

\section{Introduction}

The problem of finding the eigenvalues and the eigenvectors of a given Hamiltonian is, quite often, the first difficulty one  meets when analyzing some specific quantum system. Apart from the simple examples one can find in the textbooks, the possibility of solving this problem usually decreases fast when the system becomes more complicated. Durng the years, many techniques have been proposed and adopted to produce new solvable Hamiltonians, from the use of interwining operators, \cite{intop1,intop2,intop3}, to the so-called supersymmetric quantum mechanics, \cite{CKS,jun}, just to cite two exciting strategies. Another well-known possibility makes use of ladder operators, which are very old objects  appearing already in the analysis of the harmonic oscillator or, in  second quantization and in elementary particles, to deal with bosons and with fermions. There are thousands of books and papers dealing with bosonic and with fermionic operators, and we only refer to \cite{roman}. Ladder operators also exist in other contexts, like in many models driven by non self-adjoint Hamiltonians, see \cite{mosta,benbook,baginbagbook}, just to cite a few. The key aspect, in this case, is that the lowering and the raising operators are not one the adjoint of the other.

Some years ago, Fernandez started to set up what he  called an {\em algebraic treatment} of different quadratic Hamiltonians, not necessarily self-adjoint, \cite{fern1,fern2,fern3}. The main idea is that, given some Hamiltonian $H$, if one can find an operator $Z$ such that $[H,Z]=\lambda Z$, many interesting results can be deduced. In particular, $Z$ turns out to be a ladder operator, meaning with this that the action of powers of $Z$ on some {\em seed} eigenstate of $H$, $\hat\varphi$, can produce other eigenvectors of $H$, corresponding to different eigenvalues. More important, if $\hat\varphi$ is chosen in a proper way, we can find all the eigenvectors of $H$. Fernandez's main interest was (and still is) in concrete applications of these ladder operators. In this paper, other than considering specific quantum systems, we will also propose an abstract settings for operators of this kind, and we will extend them to other relevant situations which are interesting in quantum mechanics. In particular, we will concentrate on the existence of operators commuting with the Hamiltonian. Because of the general analysis considered in this paper, we will call all operators $Z$ satisfying the above commutation rule, or their generalized versions proposed later,  {\em abstract ladder operators}, ALOs.

The paper is organized as follows: in Section \ref{sect2} we begin our analysis assuming that an ALO  $Z$ satisfies  $[H_0,Z]=\lambda Z$, for some number $\lambda$. Here $H_0$ is a self-adjoint Hamiltonian. The results we will deduce are in line with those deduced by Fernandez in his papers, but considered here at a somehow more abstract level. In Section  \ref{sect3} we replace condition $[H_0,Z]=\lambda Z$ with $[H_0,Z]=\lambda\, Z [Z^\dagger,Z]$, which turns ot to be a good choice to enlarge significantly the class of physical systems to which our strategy applies. In Section \ref{sect4} we concentrate on a different situation, i.e. to the case of non self-adjoint Hamiltonians. We also assume that several such ALOs could exist, so that we consider the commutation rule $[H,Z_j]=\lambda_j Z_j$, $j\geq1$, $H\neq H^\dagger$, in general. Many examples are given in terms of pseudo-bosons, \cite{baginbagbook,bagrev}. Our conclusions are discussed in Section \ref{sectconcl}. To make the paper more readable, we also include a short appendix which contains few facts on the algebraic settings used almost everywhere in the following, which is relevant since most of the times the operators considered in our analysis are unbounded.

\section{ALOs for $H_0=H_0^\dagger$, pt. 1}\label{sect2}

Let $H_0$ be a self-adjoint operator acting on an Hilbert space $\Hil$, with scalar product $\langle.,.\rangle$ and related norm $\|.\|=\sqrt{\langle.,.\rangle}$. Let $Z$ be a second operator on $\Hil$ satisfying the following  equality:
\be
[H_0,Z]=\lambda Z,
\label{21}\en
for some $\lambda\in\mathbb{C}$. This equality can be satisfied by bounded or by unbounded operators. For instance, let $c$ be a lowering fermionic operator: $\{c,c^\dagger\}=cc^\dagger+c^\dagger c=\1_f$, where $\1_f$ is the identity operator in the fermionic Hilbert space $\Hil_f=\mathbb{C}^2$, and $c^2=0$. Now, if $H_0=\omega c^\dagger c$, $\omega\in\mathbb{R}$, it follows that $[H_0,c]=-\omega c$ which is exactly equation (\ref{21}) upon identifying $\lambda=-\omega$ and $Z=c$. Analogously, let $a$ be a lowering bosonic operator: $[a,a^\dagger]=\1_b$, where $\1_b$ is the identity operator in the bosonic Hilbert space $\Hil_b$, which is well known to be necessarily infinite dimensional. In particular, depending on the representation we adopt for $a$, we have $\Hil_b=\Lc^2(\mathbb{R})$ or $\Hil_b=l^2(\mathbb{N})$. Now, if $H_0=\omega a^\dagger a$, $\omega\in\mathbb{R}$, it follows that $[H_0,a]=-\omega a$ which is again equation (\ref{21}) with the identification $\lambda=-\omega$, $Z=a$. We will consider other possibilities later in the paper. When the operators involved in (\ref{21}) are bounded, as in the fermionic example, the natural operatorial settings to work with is the algebra $B(\Hil_f)$ of the bounded operators on $\Hil_f$. On the other hand, if $H_0$ or $Z$, or both, are unbounded, as for bosons, $B(\Hil_b)$  is not a good choice, since domain problems may quite easily appear. This problem is well known, and it is discussed in many papers and textbooks, where alternative algebraic frameworks are proposed to deal with certain unbounded operators. We refer to \cite{aitbook,bagrev2007,schu,trrev} and references therein for a detailed analysis of this aspect. We also refer to the Appendix, were a very coincise introduction to the algebra $\Lc^\dagger(\D)$ is given. Here it is sufficient to say that $\D$ is some suitable dense subset of $\Hil_b$ (for bosons), and $\Lc^\dagger(\D)$ is a $O^*$-algebra, see Appendix. In particular, since both $B(\Hil_b)$, or $B(\Hil)$ more in general, and $\Lc^\dagger(\D)$ are *-algebras, we can safely multiply their elements, take their adjoints, compute commutators, powers, and so on, without leaving the algebra, which is the relevant aspect for us.

For this reason, from now on we will always assume that the operators involved in our construction belong either to  $B(\Hil)$ or to $\Lc^\dagger(\D)$, for some given, or properly chosen, $\D$. We will say more on this aspect of our analysis later on.

Going back to (\ref{21}) it is first easy to check that this is equivalent to 
\be
[H_0,Z^\dagger]=-\overline{\lambda}\, Z^\dagger.
\label{22}\en
Moreover, a simple induction argument shows also that
\be
[H_0,Z^n]=n\lambda Z^n, \qquad [H_0,{Z^\dagger}^n]=-n\overline{\lambda}\, {Z^\dagger}^n, 
\label{23}\en
$n=0,1,2,\ldots$. From (\ref{21}) and (\ref{22}) it follows that
\be
[H_0,Z^\dagger Z]=(\lambda-\overline{\lambda})\, Z^\dagger Z, \qquad [H_0,ZZ^\dagger]=(\lambda-\overline{\lambda})\, ZZ^\dagger,
\label{24}\en
which, of course, implies that $[H_0,Z^\dagger Z]=[H_0,ZZ^\dagger]=0$ if $\lambda$ in (\ref{21}) is real. Hence, if $\lambda\in\mathbb{R}$, we also find that $[H_0,[Z,Z^\dagger]]=0$. 

The reality of $\lambda$ is a natural requirement, due to the fact that $H_0$ is self-adjoint. This is evident from what follows: let us assume that a nonzero eigenvector of $H_0$ exists, $\Phi_E\in\Hil$, such that \be H_0\Phi_E=E\Phi_E.\label{25}\en Of course, since $H_0=H_0^\dagger$, $E\in\mathbb{R}$. We define the following vectors
\be
\Phi_{E,n}^\downarrow:=(Z^\dagger)^n\Phi_E,\qquad  \Phi_{E,n}^\uparrow:=Z^n\Phi_E,
\label{26}\en
$n=0,1,2,3,\ldots$. We see that $\Phi_{E,0}^\downarrow=\Phi_{E,0}^\uparrow=\Phi_{E}$. It must be clarified that we are not assuming that all the vectors in (\ref{26}) are non zero. In particular, from (\ref{26}) we see the following: if some $n_0>0$ exists such that, for instance, $\Phi_{E,n_0}^\downarrow=0$, then $\Phi_{E,n}^\downarrow=0$ for all $n\geq n_0$. Analogously, if some $m_0>0$ exists such that  $\Phi_{E,m_0}^\uparrow=0$, then $\Phi_{E,m}^\uparrow=0$ for all $m\geq m_0$.

\begin{thm}\label{thm1}
With the above definitions the following results hold: if $\lambda\in\mathbb{R}$, then

(1) if $\Phi_{E,n}^\uparrow\neq 0$, then $\Phi_{E,n}^\uparrow$ is an eigenstate of $H_0$ and
\be
H_0\Phi_{E,n}^\uparrow=(E+n\lambda)\Phi_{E,n}^\uparrow.
\label{27}\en

(2) if $\Phi_{E,n}^\downarrow\neq 0$, then $\Phi_{E,n}^\downarrow$ is an eigenstate of $H_0$ and
\be
H_0\Phi_{E,n}^\downarrow=(E-n\lambda)\Phi_{E,n}^\downarrow.
\label{28}\en

If $\lambda\in\mathbb{C}\setminus\mathbb{R}$, then $\Phi_{E,n}^\uparrow=\Phi_{E,n}^\downarrow=0$, for all $n\geq1$.

\end{thm}

\begin{proof}
The proof of (\ref{27}) is based on formula (\ref{23}):
$$
H_0\Phi_{E,n}^\uparrow=H_0Z^n\Phi_{E}=\left([H_0,Z^n]+Z^nH_0\right)\Phi_E=\left(n\lambda Z^n+Z^nH_0\right)\Phi_E=(n\lambda+E)Z^n\Phi_E.
$$
Equation (\ref{28}) can be deduced in a similar way, by  using also the fact that $\lambda\in\mathbb{R}$. Otherwise, rather than (\ref{28}), we would get $H_0\Phi_{E,n}^\downarrow=(E-n\overline{\lambda})\Phi_{E,n}^\downarrow$. But, being $H_0=H_0^\dagger$, all its eigenvalues must be real, which is not possible if the imaginary part of $\lambda$ is not zero. If this is the case, then we conclude that $\|\Phi_{E,n}^\uparrow\|=\|\Phi_{E,n}^\downarrow\|=0$, for all $n\geq1$, and our claim follows.

\end{proof}

\vspace{2mm}

{\bf Remark:--} From now on, except when stated differently, we will assume that $\lambda$ is real. This is because, according to the previous theorem, this is the only way to  deal with vectors which are  not necessarily zero. 

\vspace{2mm} 

If $\lambda>0$, Theorem \ref{thm1} states that $Z$ is a raising operator, while $Z^\dagger$ is a lowering operator. Of course, the vectors in (\ref{26}) obey some orthogonality conditions. In particular we find that
\be
\langle\Phi_{E,n}^\downarrow,\Phi_{E,m}^\downarrow\rangle=\langle\Phi_{E,n}^\uparrow,\Phi_{E,m}^\uparrow\rangle=0,
\label{29}\en
whenever $n\neq m$. Also,
\be
\langle\Phi_{E,n}^\downarrow,\Phi_{E,m}^\uparrow\rangle=0,
\label{210}\en
for all $n+m>0$. These are all consequences of the fact that these vectors are eigenstates of $H_0$, corresponding to different eigenvalues.

Suppose now that the eigenvalues of $H_0$ are all non degenerate. We recall again that $\lambda$ is supposed to be real. In this case we can check that all non zero $\Phi_{E,n}^\downarrow$ and $\Phi_{E,n}^\uparrow$ satisfy other eigenvalue equations. For instance, since $H_0$ commutes with $Z^\dagger Z$, it is clear that $Z^\dagger Z\Phi_E$, if is not zero, is again an eigenstate of $H_0$ with eigenvalue $E$. Then, since the multiplicity of $E$ is one, $m(E)=1$, a real number $\mu_1$ exists such that $Z^\dagger Z\Phi_E=\mu_1\Phi_E$. Since $Z^\dagger Z$ is a non negative operator, $\mu_1\geq0$. In particular, $\mu_1=0$ if and only if $\Phi_{E,1}^\uparrow=Z\Phi_E=0$. This is because, taking the scalar product of $Z^\dagger Z\Phi_E=\mu_1\Phi_E$ with $\Phi_E$, we get $\mu_1\|\Phi_E\|^2=\|Z\Phi_E\|^2=\|\Phi_{E,1}^\uparrow\|^2$. This equality also shows that, being $\Phi_E\neq0$, $\mu_1=\frac{\|\Phi_{E,1}^\uparrow\|^2}{\|\Phi_{E}\|^2}$, which is positive, as already stated. The same procedure can be repeated for other eigenvalues: if $Z^\dagger Z\Phi_{E,1}^\uparrow$, is not zero, it is an eigenstate of $H_0$ with eigenvalue $E+\lambda$. Then, since  $m(E+\lambda)=1$, a real number $\mu_2$ exists such that $Z^\dagger Z\Phi_{E,1}^\uparrow=\mu_2\Phi_{E,1}^\uparrow$. Again, $\mu_2=0$ if and only if $\Phi_{E,2}^\uparrow=0$, which cannot be the case here, because of our assumption on  $Z^\dagger Z\Phi_{E,1}^\uparrow$. Hence $\mu_2=\frac{\|Z\Phi_{E,1}^\uparrow\|^2}{\|\Phi_{E,1}^\uparrow\|^2}$, which is obviously positive. Iterating this procedure, we find that, if $Z\Phi_{E,n-1}^\uparrow\neq0$ for a given $n$, then $Z\Phi_{E,k-1}^\uparrow\neq0$ for all $k\leq n$ and we have
$$
Z^\dagger Z \Phi_{E,n-1}^\uparrow=\mu_n \Phi_{E,n-1}^\uparrow, \qquad\mbox{ with }\qquad  \mu_n=\frac{\|Z\Phi_{E,n-1}^\uparrow\|^2}{\|\Phi_{E,n-1}^\uparrow\|^2}=\frac{\|\Phi_{E,n}^\uparrow\|^2}{\|\Phi_{E,n-1}^\uparrow\|^2}.
$$
The norm of the vectors $ \Phi_{E,n}^\uparrow$ can be rewritten in terms of the $\mu_n$'s:
$$
\|\Phi_{E,n}^\uparrow\|^2=\mu_n\,\mu_{n-1}\cdots\mu_1\|\Phi_{E,0}^\uparrow\|^2.
$$
Moreover, we can deduce that $\mu_m=0$ for some $m>0$, $\mu_{m-1}\neq0$, if and only if $\Phi_{E,k}^\uparrow=0$, for all $k\geq m$. This implies, in view of Theorem \ref{thm1}, that the point spectrum of $H_0$, $\sigma_p(H_0)$, \cite{rs}, is bounded from above by $E_{max}:=E+m\lambda$.

These results, other than for $Z^\dagger Z$, can be restated for $ZZ^\dagger$, since this operator also commutes with $H_0$ when $\lambda\in\mathbb{R}$. We have the following:

if $Z^\dagger\Phi_{E,n-1}^\downarrow\neq0$ for a given $n$, then $Z^\dagger\Phi_{E,k-1}^\downarrow\neq0$ for all $k\leq n$ and we have
$$
ZZ^\dagger \Phi_{E,n-1}^\downarrow=\nu_n \Phi_{E,n-1}^\downarrow, \qquad\mbox{ with }\qquad  \nu_n=\frac{\|Z^\dagger\Phi_{E,n-1}^\downarrow\|^2}{\|\Phi_{E,n-1}^\downarrow\|^2}=\frac{\|\Phi_{E,n}^\downarrow\|^2}{\|\Phi_{E,n-1}^\downarrow\|^2}.
$$
Also:
$$
\|\Phi_{E,n}^\downarrow\|^2=\nu_n\,\nu_{n-1}\cdots\nu_1\|\Phi_{E,0}^\downarrow\|^2,
$$
and $\nu_m=0$ for some $m>0$, $\nu_{m-1}\neq0$, if and only if $\Phi_{E,k}^\downarrow=0$, for all $k\geq m$. This implies, again in view of Theorem \ref{thm1}, that  $\sigma_p(H_0)$ is bounded from below by $E_{min}:=E-m\lambda$. This is the case, for instance, of any non negative Hamiltonian, $H_0\geq0$.

\vspace{2mm}

{\bf Remark:--} It is possible to deduce some relations between $\mu_n$ and $\nu_n$. However, these relations become more and more complicated when $n$ increases and are not so useful. For instance, it is easy to see that
$$
\mu_1-\nu_1=\frac{\langle[Z^\dagger,Z]\Phi_E,\Phi_E\rangle}{\langle\Phi_E,\Phi_E\rangle}.
$$

\vspace{2mm}

We have already shown that bosonic and fermionic operators produce examples of this functional settings, at least when the Hamiltonian $H_0$ of the system is proportional to a number operator. Here we don't insist on examples of this kind,  because several interesting examples have already been considered by Fernandez along the years and since we are more interested in extending these results  to other situations not covered so far. In particular, we will consider the case in which several operators $Z_j$ exist which satisfy a commutation rule like the one in (\ref{21}), opening also to the possibility that the Hamiltonian is not self-adjoint. Another extension, discussed in the next section, generalize (\ref{21}) in such a way that all the factorizable Hamiltonians, and those written in terms of generalized Heisenberg algebras, \cite{bcg,curado1,curado2}, fit into the scheme.

\section{ALOs for $H_0=H_0^\dagger$, pt. 2}\label{sect3}

Let us assume that an Hamiltonian $H_0=H_0^\dagger$ and $Z$, operators on $\Hil$ as in the previous section, obey the following commutation rule:
\be
[H_0,Z]=\lambda\, Z [Z^\dagger,Z],
\label{31}\en
for some real number $\lambda$. Here, except when stated, to fix the ideas we will assume that $\lambda$ is strictly positive. As in Section \ref{sect2}, formula (\ref{31}) makes sense for bounded $H_0$ and $Z$, or when they both belong to $\Lc^\dagger(\D)$ for some dense $\D$ in $\Hil$.  We notice that formula (\ref{31}) returns (\ref{21}) when $[Z,Z^\dagger]$ is the identity operator (or, with some redefinition of the quantities, is proportional to the identity operator), but is different from (\ref{21}) otherwise. We also notice that, if $Z=Z^\dagger$ (or, more in general, if $Z$ is normal), then (\ref{31}) implies that $[H_0,Z]=0$.

Formula (\ref{31}) is equivalent to 
\be
[H_0,Z^\dagger]=-\lambda\,  [Z^\dagger,Z]Z^\dagger.
\label{32}\en
Moreover, (\ref{31}) and (\ref{32}) can be extended as follows:
\be
[H_0,Z^n]=\lambda\, Z [Z^\dagger,Z^n], \qquad [H_0,{Z^\dagger}^n]=-\lambda\,  [Z,{Z^\dagger}^n]Z^\dagger,
\label{33}\en
for all $n\geq0$. An interesting difference with respect to the results in Section \ref{sect2} is that, in general, $Z^\dagger Z$ and $ZZ^\dagger$ behave differently regarding their commutativity with $H_0$. In fact, using (\ref{31}) and (\ref{32}), we can easily check that
\be
[H_0,ZZ^\dagger]=0, 
\label{34}\en
while for $[H_0,Z^\dagger Z]$ we deduce the following equivalent expressions, in general different from zero:
\be[H_0,Z^\dagger Z]=\lambda[ZZ^\dagger,Z^\dagger Z]= [H_0,[Z^\dagger,Z]]=\lambda[Z^\dagger Z,[Z^\dagger,Z]].
\label{35}\en

Let us now consider some interesting example satisfying (\ref{31}) and, consequently, the other equalities deduced above.

{\bf A first class of Examples: $H_0$ is factorizable. } Suppose that $H_0$ is self-adjoint (as always in Sections \ref{sect2} and \ref{sect3}) and factorizable. This means that it can be written as $H_0=\omega A^\dagger A$, for some operator $A$ and some real $\omega$. Of course, a complex $\omega$ or a factorization of the form $BA$, $B\neq A^\dagger$, would not be compatible with the condition  $H_0=H_0^\dagger$. In this case, if we put $Z=A^\dagger$ (independently of the explicit expression of $A$) we deduce that
$$
[H_0,Z]=\omega[A^\dagger A,A^\dagger]=\omega A^\dagger [A,A^\dagger]=\omega Z[Z^\dagger,Z],
$$
which is exactly formula (\ref{31}) with $\lambda=\omega$.

Of course, if $A=a$, where $[a,a^\dagger]=\1_b$, and $H_0=\omega a^\dagger a$, see Section \ref{sect2}, we are exactly in this situation. However, since $[Z^\dagger,Z]=[a,a^\dagger]=\1_b$, formula (\ref{31}) simplifies and returns formula (\ref{21}).

Formula (\ref{31}) is also satisfied if $A=c$, where $\{c,c^\dagger\}=\1_f$, and $H_0=\omega c^\dagger c$, see again Section \ref{sect2}, since $H_0$ is factorized. In this case, however, since $[Z^\dagger,Z]=[c,c^\dagger]=2cc^\dagger-\1_f$, formula (\ref{31}) does not return formula (\ref{21}).

Another example of the same kind which satisfies (\ref{31}) but not (\ref{21}) can be constructed from quons, \cite{andr,bagquons,fiv,green,kar,moh}. In this case we have ladder operators, $b$ and $b^\dagger$, satisfying the $q$-mutation relation
$$
bb^\dagger-qb^\dagger b=\1,
$$
where $q\in[-1,1]$ and $\1$ is the identity operator in the {\em quonic} Hilbert space, \cite{fiv,moh}. Once again, we take as Hamiltonian the operator $H_0=\omega b^\dagger b$. We recall that $H_0$ is not proportional to the number operator $N$ for quons, but it is still diagonal in the eigenvectors of $N$, \cite{fiv,moh}. Formula (\ref{31}) becomes, taking $Z=b^\dagger$,
$$
[H_0,Z]=[H_0,b^\dagger]=\omega b^\dagger \left(\1+b^\dagger b(q-1)\right)= \omega  Z\left(\1+ZZ^\dagger (q-1)\right),
$$
which returns exactly (\ref{21}) when $q=1$, i.e. when quons become ordinary bosons, but not for the other values of $q$.

\vspace{3mm}

{\bf Another class of examples: generalized Heisenberg algebra.} In \cite{curado1,curado2}, and references therein, the notion of {\em generalized Heisenberg algebra } has been proposed and analyzed. This is based on two operators, $d$ and $H_0=H_0^\dagger$, and a suitable increasing real function $f(x)$, satisfying the equalities
\be
d\,H_0=f(H_0)d, \qquad [d,d^\dagger]=f(H_0)-H_0.
\label{36}\en
From the first equation we deduce $H_0d^\dagger=d^\dagger f(H_0)$. In \cite{curado1,curado2} some examples of this algebraic settings have been proposed, and others have been considered in \cite{bcg} for $H_0$ not necessarily self-adjoint. For instance, the P\"oschl–Teller or the infinite square well potentials can be considered in this perspective.  Now, independently of the explicit form of $H_0$, if we take $Z=d^\dagger$ we have
$$
[H_0,Z]=H_0d^\dagger-d^\dagger H_0=d^\dagger\left(f(H_0)-H_0\right)=d^\dagger[d,d^\dagger]=Z[Z^\dagger,Z],
$$
which is exactly formula (\ref{31}) with $\lambda=1$. We observe that, since $[H_0,ZZ^\dagger]=0$,
$$
[H_0,Z^\dagger Z]=[H_0,[Z^\dagger,Z]]=[H_0,[d,d^\dagger]]=[H_0,f(H_0)-H_0]=0.
$$

\vspace{3mm}

In all examples considered above it is easy to check that, as expected, $[H_0,ZZ^\dagger]=0$. Interestingly enough, we can also check (as we did in the last example) that $[H_0,Z^\dagger Z]=0$, so that (\ref{35}) simplifies significantly. This shows that the main assumption of  Theorem \ref{thm2} below holds true in many situations, even if does not appear to be completely general. Before stating the theorem we introduce, as we did in Section \ref{sect2}, the following vectors:
\be
\Phi_{E,n}^\uparrow:=Z^n\Phi_E,
\label{37}\en
$n\geq0$, where $\Phi_E$ is a (nonzero) eigenstate of $H_0$ with eigenvalue $E$, see (\ref{25}): $H_0\Phi_E=E\Phi_E$. Of course, $\Phi_{E,0}^\uparrow=\Phi_E\neq0$. However, it might happen in principle that a certain $n_0>0$ does exist such that  $\Phi_{E,n_0}^\uparrow=0$, with $\Phi_{E,n_0-1}^\uparrow\neq0$. This implies that $\Phi_{E,n}^\uparrow=0$ for all $n\geq n_0$. Of course, one might wonder why we are not considering here the vectors $\Phi_{E,n}^\downarrow:={Z^\dagger}^n\Phi_E$. We will clarify the reason for this later.

\begin{thm}\label{thm2}
Let us assume that all the eigenvalues of $H_0$ are non degenerate, and that $[H_0,Z^\dagger Z]=0$. Then, for all $n\geq0$ such that $\Phi_{E,n}^\uparrow\neq0$, we can introduce
\be
\mu_{E,n}:=\frac{\langle\Phi_{E,n}^\uparrow,[Z^\dagger,Z]\Phi_{E,n}^\uparrow\rangle}{\|\Phi_{E,n}^\uparrow\|^2},
\label{38}\en
and we have
\be
[Z^\dagger,Z]\Phi_{E,n}^\uparrow=\mu_{E,n}\Phi_{E,n}^\uparrow.
\label{39}\en
Moreover, calling $E_n^\uparrow=\lambda\sum_{k=0}^{n-1}\mu_{E,k}+E$, $n\geq1$,  and $E_0^\uparrow=E$, we get
\be
H_0\Phi_{E,n}^\uparrow=E_n^\uparrow\,\Phi_{E,n}^\uparrow.
\label{310}\en

\end{thm}

\begin{proof}
	First of all, it is clear that $H_0$ commutes with $[Z^\dagger,Z]$. We now use induction to prove our claim starting with $n=0$. In this case we have already commented that $\Phi_{E,0}^\uparrow=\Phi_E\neq0$. Hence (\ref{310}) is clearly satisfied with $E_0^\uparrow=E$. Moreover, since $[H_0,[Z^\dagger,Z]]=0$, we deduce that $[Z^\dagger,Z]\Phi_{E,0}^\uparrow$ is an eigenstate of $H_0$ with eigenvalue $E$. But all the eigenvalues of $H_0$ have multiplicity one. Therefore, a (in principle) complex number $\mu_{E,0}$ exists such that
	$$
	[Z^\dagger,Z]\Phi_{E,0}^\uparrow=\mu_{E,0}\Phi_{E,0}^\uparrow.
	$$
	Now, taking the scalar product of this equality with $\Phi_{E,0}^\uparrow$, we deduce formula (\ref{38}) with $n=0$. Incidentally, it is clear that $\mu_{E,0}$ is real, even if nothing can be said a priori on its sign.
	
	Let us assume now that, for a given $k\geq0$, $\Phi_{E,k}^\uparrow\neq0$ and formulas (\ref{38})-(\ref{310}) are satisfied. We will now prove that, if also $\Phi_{E,k+1}^\uparrow\neq0$, similar formulas can be deduced. First we notice that $\Phi_{E,k+1}^\uparrow=Z\Phi_{E,k}^\uparrow$. Hence we have, using (\ref{31}),
	$$
	H_0\Phi_{E,k+1}^\uparrow=\left([H_0,Z]+ZH_0\right)\Phi_{E,k}^\uparrow=\left(\lambda\, Z [Z^\dagger,Z]+ZH_0\right)\Phi_{E,k}^\uparrow.
	$$
	But, because of the induction assumption, we have
	$$
	[Z^\dagger,Z]\Phi_{E,k}^\uparrow=\mu_{E,k}\Phi_{E,k}^\uparrow, \quad \mbox{ and }\quad H_0\Phi_{E,k}^\uparrow=E_k^\uparrow\Phi_{E,k}^\uparrow.
	$$
	Therefore,
	$$
	H_0\Phi_{E,k+1}^\uparrow=\left(\lambda\mu_{E,k}+E_{k}^\uparrow\right)Z\Phi_{E,k}^\uparrow=E_{k+1}^\uparrow\Phi_{E,k+1}^\uparrow,
	$$
	which is what we had to check. Here we have used the following equality:
	$$
	\lambda\,\mu_{E,k}+E_{k}^\uparrow=\lambda\,\mu_{E,k}+\left(\lambda\sum_{j=0}^{k-1}\mu_{E,j}+E\right)=\lambda\sum_{j=0}^{k}\mu_{E,j}+E=E_{k+1}^\uparrow.
	$$ 
	Formulas (\ref{38}) and (\ref{39}) for $\Phi_{E,k+1}^\uparrow$ can now be deduced easily.

\end{proof}

\vspace{2mm}

{\bf Remarks:--} (1) This theorem shows that the non zero $\Phi_{E,n}^\uparrow$ are eigenstates of $H_0$ and of $[Z^\dagger,Z]$. In fact, they are also eigenstates of, separately, $ZZ^\dagger$ and $Z^\dagger Z$:
\be
ZZ^\dagger\Phi_{E,n}^\uparrow=\alpha_{E,n}\Phi_{E,n}^\uparrow, \qquad Z^\dagger Z\Phi_{E,n}^\uparrow=\beta_{E,n}\Phi_{E,n}^\uparrow,
\label{311}\en
with
$$
\alpha_{E,n}=\frac{\|Z^\dagger\Phi_{E,n}^\uparrow\|^2}{\|\Phi_{E,n}^\uparrow\|^2}, \qquad \beta_{E,n}=\frac{\|Z\Phi_{E,n}^\uparrow\|^2}{\|\Phi_{E,n}^\uparrow\|^2}.
$$
Hence
\be
\mu_{E,n}=\beta_{E,n}-\alpha_{E,n}.
\label{312}\en

\vspace{2mm}

(2) Of course, in what deduced in the theorem, it is essential to work with non zero vectors. If, for some integer $m$, $\Phi_{E,m}^\uparrow=0$, most of the results we have deduced should be reconsidered, since they can make no sense, mainly because the zero vector cannot be the eigenvector of any operator. This is what happens, in particular, if $H_0$ is bounded from above or from below.

\vspace{2mm}

Definition (\ref{37}) implies that $Z$ is a raising operator\footnote{ $Z$ should be called a raising operator if the energy of the state $\Phi_{E,n+1}^\uparrow=Z\Phi_{E,n}^\uparrow$ is higher than that of $\Phi_{E,n}^\uparrow$. This is not so clear for us, since we cannot say much on the sign of $\mu_{E,n}$. Nevertheles, we keep this terminology since $Z$ increases the quantum number $n$.}. An interesting consequence of Theorem \ref{thm2} and of formulas (\ref{311}) is that $Z^\dagger$ is a lowering operator. In fact, because of (\ref{37}), we have $Z^\dagger Z\Phi_{E,n}^\uparrow=Z^\dagger \Phi_{E,n+1}^\uparrow=\beta_{E,n}\Phi_{E,n}^\uparrow$. Summarizing we have
\be
Z\Phi_{E,n}^\uparrow=\Phi_{E,n+1}^\uparrow, \qquad Z^\dagger\Phi_{E,n}^\uparrow=\beta_{E,n-1}\Phi_{E,n-1}^\uparrow,
\label{313}\en
whenever $Z^\dagger\Phi_{E,n}^\uparrow\neq0$. It is maybe useful to notice that this is not so evident a priori, because (\ref{31}) introduces a sort of asymmetry between $Z$ and $Z^\dagger$. Luckily enough, this asymmetry is only reflected by the difference of the coefficients in formula (\ref{313}). It might be useful to notice that formula (\ref{313}), as well as all the other results in this section, could be adapted to the case considered in Section \ref{sect2} assuming that the ALO $Z$ satisfies $[Z^\dagger,Z]=\1$. In this case, in fact, (\ref{31}) collapses into (\ref{21}). 

\vspace{3mm}

Comparing Section \ref{sect2} and Section \ref{sect3}, it is clear that what is missing here are the {\em down-arrow vectors} $\Phi_{E,n}^\downarrow=(Z^\dagger)^n\Phi_E$, see (\ref{26}). The reason is that there is no analogous version of Theorem \ref{thm2}, and this is a simple consequence of the particular form of the commutation rule (\ref{31}). For instance, in order to check if $\Phi_{E,1}^\downarrow=Z^\dagger\Phi_E$ is an eigenstate of $H_0$ (assuming it is not zero), we compute the following:
$$
H_0\Phi_{E,1}^\downarrow=\left([H_0,Z^\dagger]+Z^\dagger H_0\right)\Phi_E=\left(-\lambda\,  [Z^\dagger,Z]Z^\dagger+Z^\dagger H_0\right)\Phi_E,
$$
using (\ref{32}). Now, while it is easy to compute the second contribution in the RHS, $Z^\dagger H_0\Phi_E=E\Phi_{E,1}^\downarrow$, it is not clear what would be the result of $[Z^\dagger,Z]Z^\dagger\Phi_E$, even in the case of all eigenvalues of $H_0$ with multiplicity one. Hence we cannot conclude, without adding more conditions, that $\Phi_{E,1}^\downarrow$ is an eigenstate of $H_0$.

 However, as we have seen before, $Z^\dagger$ acts as a lowering operator when acting on the other vectors, $\Phi_{E,n}^\uparrow$.

\section{ALOs for $H\neq H^\dagger$}\label{sect4}

From now on we will no longer require that the Hamiltonian is self-adjoint, except if stated explicitly: $H\neq H^\dagger$. Moreover, we will assume that $H$ admits more than just a single ladder operator. Then we assume that:
\be
[H,Z_j]=\lambda_j Z_j,
\label{41}\en
$\lambda_j\in\mathbb{C}$, for $j=1,2,\ldots,N$, or, equivalently, that  
\be
[H^\dagger,Z_j^\dagger]=-\overline{\lambda_j}\, Z_j^\dagger.
\label{42}\en
Of course, since $H\neq H^\dagger$, there is no reason to require that $\lambda_j\in\mathbb{R}$. However, it is clear that these commutators return those in (\ref{21}) and (\ref{22}) if $N=1$ and $H=H_0=H_0^\dagger$.

Formula (\ref{23}) can be deduced also in this situation, and we get:
\be
[H,Z_j^n]=n\lambda_j Z_j^n, \qquad [H^\dagger,{Z_j^\dagger}^n]=-n\overline{\lambda_j}\, {Z_j^\dagger}^n, 
\label{43}\en
for all $n=0,1,2,\ldots$ and $j=1,2,\ldots,N$. 

Let us suppose that two, in general different, nonzero vectors exist, $\varphi_E$ and $\psi_E$, such that
\be
H\varphi_E=E\varphi_E, \qquad H^\dagger\psi_E=\overline{E}\,\psi_E.
\label{44}\en
These are respectively eigenstates of $H$ and $H^\dagger$, with complex conjugate eigenvalues. We define the following vectors
\be
\varphi_{E:j,n}=Z_j^n\varphi_E, \qquad \psi_{E:j,n}={Z_j^\dagger}^n\psi_E,
\label{45}\en
for all $n=0,1,2,\ldots$ and $j=1,2,\ldots,N$. Of course, if some integer $n_0(j)$ exists such that $\varphi_{E:j,n_0(j)}=0$, it follows that $\varphi_{E:j,n}=0$ for all $n\geq n_0(j)$. Similarly, if some integer $m_0(j)$ exists such that $\psi_{E:j,m_0(j)}=0$,  it follows that $\psi_{E:j,m}=0$ for all $m\geq m_0(j)$.

The following result shows that all nonzero $\varphi_{E:j,n}$ and  $\psi_{E:j,n}$ are eigenstates of $H$ and $H^\dagger$. In particular we have
\be
H\varphi_{E:j,n}=\epsilon_{E:j,n}\varphi_{E:j,n}, \qquad H^\dagger\psi_{E:j,n}=\overline{\epsilon_{E:j,-n}}\,\psi_{E:j,n},
\label{46}\en
where
\be
\epsilon_{E:j,n}=E+n\lambda_j.
\label{47}\en
The proof is not significantly different from that of (\ref{27}) and (\ref{28}), and will not be repeated. Calling $\F_\varphi=\{\varphi_{E:j,n}:\, j=1,2,\ldots, N; \,n\geq0\}$ and $\F_\psi=\{\psi_{E:j,n}:\, j=1,2,\ldots, N; \,n\geq0\}$, these two sets are biorthogonal in the following sense:
\be
\mbox{if } n\lambda_j\neq -m\lambda_l \quad\Rightarrow\quad\langle\psi_{E:l,m},\varphi_{E:j,n}\rangle=0.
\label{48}
\en
This is a consequence of (\ref{46}) and of the fact that $\epsilon_{E:j,n}\neq \epsilon_{E:l,-m}$ if and only if $n\lambda_j\neq -m\lambda_l$. This result is not surprising since we are taking the scalar product of eigenvectors of $H$ and $H^\dagger$, which are biorthogonal sets under very general situations. Notice that, in principle, these operators are not isospectrals, except that in particular cases. For instance, this happens when $E$ and $\lambda_j$ are real, and when for each $j=1,2,\ldots,N$ there is a $l=1,2,\ldots,N$ such that $\lambda_j=-\lambda_l$.

In Sections \ref{sect2} and \ref{sect3} we have shown that the eigenstates of $H_0$ are also eigenstates of other operators constructed with $Z$ and $Z^\dagger$. We can extend these results also to the present settings. This is because, as it is easy to check,
\be
[H,Z_jZ_k]=\left(\lambda_j+\lambda_k\right)Z_jZ_k.
\label{49}\en
It is clear that the order of $Z_j$ and $Z_k$ does not affect this formula. In other words, if $\lambda_j+\lambda_k=0$, both $Z_jZ_k$ and $Z_kZ_j$ commute with $H$. And this true both when $[Z_j,Z_k]=0$, and when $Z_j$ and $Z_k$ do not commute. Moreover, formula (\ref{49}) can be easily extended as follows:
\be
[H,Z_{j_1}Z_{j_2}\cdots Z_{j_k}]=\left(\sum_{l=1}^k\lambda_{j_l}\right)Z_{j_1}Z_{j_2}\cdots Z_{j_k},
\label{410}\en
$k=1,2,3,\ldots$, which produces an equivalent formula for the adjoint:
\be
[H^\dagger,Z_{j_1}^\dagger Z_{j_2}^\dagger\cdots Z_{j_k}^\dagger]=-\left(\sum_{l=1}^k\overline{\lambda_{j_l}}\right)Z_{j_1}^\dagger Z_{j_2}^\dagger\cdots Z_{j_k}^\dagger,
\label{411}\en
$k=1,2,3,\ldots$. Again, the order of the operators is important in general, but if $\sum_{l=1}^k\lambda_{j_l}=0$, all the permutations of  the product $Z_{j_1}Z_{j_2}\cdots Z_{j_k}$ commute with $H$ and, similarly, all the permutations of $Z_{j_1}^\dagger Z_{j_2}^\dagger\cdots Z_{j_k}^\dagger$ commute with $H^\dagger$. 

\vspace{2mm}

{\bf Remark:--}  It is clear that formula (\ref{49}) returns (\ref{24}) when $Z_j=Z$ and $Z_k=Z^\dagger$, and when $H=H_0=H_0^\dagger$. This shows that all the results in this section could be specialized to the case considered in Section \ref{sect2}, if several ALOs exist.

\vspace{2mm}

We introduce the following 

\begin{defn}
	The operators $(H,Z_1,\ldots,Z_N)$ obeying (\ref{41}) satisfy {\em Condition $N_0$} if, for some $N_0>1$, $\sum_{l=1}^{N_0}\lambda_{j_l}=0$.
\end{defn}

Hence: if $(H,Z_1,\ldots,Z_N)$ satisfy {\em Condition $N_0$}, $H$ commutes with $Z_{j_1} Z_{j_2}\cdots Z_{j_{N_0}}$ and with any its permutation. Moreover, $H^\dagger$ commutes with $Z_{j_1}^\dagger Z_{j_2}^\dagger\cdots Z_{j_{N_0}}^\dagger$ and with any its permutation. It is now easy to conclude that, under  Condition $N_0$, and assuming that all the eigenvalues of $H$ have multiplicity one, $\varphi_{E:j,n}$ is an eigenstate of $Z_{j_1} Z_{j_2}\cdots Z_{j_{N_0}}$ and $\psi_{E:j,n}$ is an eigenstate of $Z_{j_1}^\dagger Z_{j_2}^\dagger\cdots Z_{j_{N_0}}^\dagger$, together with all their permutations. In other words we have
\be
Z_{j_1} Z_{j_2}\cdots Z_{j_{N_0}}\varphi_{E:j,n}=z_{E:j,n}(j_1,j_2,\ldots,j_{N_0})\varphi_{E:j,n}
\label{412}\en
and 
\be
Z_{j_1}^\dagger Z_{j_2}^\dagger\cdots Z_{j_{N_0}}^\dagger\psi_{E:j,n}=w_{E:j,n}(j_1,j_2,\ldots,j_{N_0})\psi_{E:j,n}
\label{413}\en
for some complex numbers $z_{E:j,n}(j_1,j_2,\ldots,j_{N_0})$ and $w_{E:j,n}(j_1,j_2,\ldots,j_{N_0})$. 

\vspace{2mm}

{\bf Remark:--} The presence of $j_1,j_2,\ldots,j_{N_0}$ in the brackets of $z_{E:j,n}$ and $w_{E:j,n}$ is important. Suppose, to simplify the situation, that $N_0=2$. This means that $\varphi_{E:j,n}$ is an eigenstate of both $Z_1Z_2$ and of $Z_2Z_1$: \be Z_1Z_2\varphi_{E:j,n}=z_{E:j,n}(j_1,j_2)\varphi_{E:j,n}, \qquad Z_2Z_1\varphi_{E:j,n}=z_{E:j,n}(j_2,j_1)\varphi_{E:j,n}.\label{414}\en
There is no reason, if $[Z_1,Z_2]\neq0$, to have $z_{E:j,n}(j_1,j_2)= z_{E:j,n}(j_2,j_1)$, In fact, we will see soon, in some concrete applications, that this is not  true.

\vspace{2mm}

We repeat once more that all the results deduced in this section can be restated for the case of self-adjoint Hamiltonians, i.e. when $H=H^\dagger$. In other words, the results discussed in this section can be seen as an extension of those of Section \ref{sect2} to the case when more than just one ALO exist.

\subsection{Examples}

This part of Section \ref{sect4} is devoted to the analysis of some examples fitting the structure proposed so far, starting with one-dimensional pseudo-bosons and then moving to a two-dimensional case. Also, applications to two different versions of 2-d non hermitian coupled harmonic oscillators will be discussed. 

\subsubsection{Pseudo-bosons in $d=1$}

The first example is related to the so-called $\D$-pseudo bosons ($\D$-PBs), in its $d=1$ version. We refer to \cite{baginbagbook} for many mathematical properties of these operators. Here we just  need to say that, see \cite{bagrus2018}, if $a$ and $b$ are operators on $\Hil$ satisfying $[a,b]f=f$ for all $f\in\D$, $\D$ a fixed subset of $\Hil$\footnote{In many one dimensional systems we have $\Hil=\Lc^2(\mathbb{R})$ and $\D=\Sc(\mathbb{R})$, the set of test functions.}, then $a$ and $b$ are good candidates for being $\D$-pseudo bosonic operators or weak pseudo bosons, \cite{wpbs}.  Indeed, as it is widely discussed in \cite{baginbagbook,wpbs}, the validity of the commutation relation $[x,y]=\1$ (in some proper sense, e.g. as unbounded operators) does not guarantee that $x$ and $y$ are $\D$-PBs. They need to satisfy some extra conditions. Noneless, as discussed in \cite{bagrus2018}, in many relevant situations it is possible to check that $a$ and $b$ are elements of a *-algebra $\Lc^\dagger(\D)$, see Appendix, and, as such,  all powers and combinations of powers of $a$ and $b$ are still elements of $\Lc^\dagger(\D)$. In particular, thought being unbounded, $[a,b]$ is a well defined element of $\LD$.

After these very minimal preliminaries, let us define an Hamiltonian $H=\omega ba=\omega N$, $\omega\in\mathbb{R}$ and $N=ba$. Similar Hamiltonians appear quite often when dealing with pseudo-Hermitian quantum mechanics, \cite{mosta}, or with PT-quantum mechanics, \cite{ben,benbook}. For instance, they appear in connection with the Swanson model, or with many non self-adjoint versions of the harmonic oscillator, \cite{swan,dapro,bagpb5,baginbagbook}.

Let us now fix $Z_1=a$ and $Z_2=b$. It is easy to see that, if we put $\lambda_1=-\lambda_2=-\omega$, condition (\ref{41}) is satisfied: $[H,Z_1]=\lambda_1 Z_1$ and $[H,Z_2]=\lambda_2 Z_2$. Condition (\ref{42}) can be also checked easily, noticing that $H^\dagger=\omega N^\dagger$. Now, since $a$ and $b$ are $\D$-PBs, two non-zero vectors exist in $\D$ such that $a\varphi_0=b^\dagger\psi_0=0$. These are the {\em seeds vectors} out of which we can construct the following (non-zero) vectors in $\D$:
$$
\varphi_n=\frac{1}{\sqrt{n!}}\,b^n\varphi_0, \qquad \psi_n=\frac{1}{\sqrt{n!}}\,{a^\dagger}^n\psi_0,
$$
$n\geq0$. They satisfy the eigenvalue equations $H\varphi_n=\omega n\varphi_n$ and $H^\dagger\psi_n=\omega n\psi_n$, for all $n\geq0$. In particular, since $H\varphi_0=H^\dagger\psi_0=0$, $E=0$ in (\ref{44}). With this in mind, the vectors in (\ref{45}) can be rewritten as follows:
$$
\varphi_{0:j,n}=Z_j^n\varphi_0=\left\{
\begin{array}{ll}
a^n\varphi_0=0\hspace{1.9cm}\, \forall n\geq 1, \quad j=1\\
b^n\varphi_0=\sqrt{n!}\,\,\varphi_{n}, \qquad \forall n\geq 0,\quad j=2\\
\end{array}
\right.
$$
and
$$
\psi_{0:j,n}={Z_j^\dagger}^n\psi_0=\left\{
\begin{array}{ll}
{a^\dagger}^n\psi_0=\sqrt{n!}\,\psi_{n}, \qquad \forall n\geq 0, \quad j=1\\
{b^\dagger}^n\psi_0=0, \hspace{1.8cm} \forall n\geq 1,\quad j=2\\
\end{array}
\right.
$$
In particular we have $\varphi_{0:j,0}=\varphi_0$ and $\psi_{0:j,0}=\psi_0$, $j=1,2$. Needless to say, these results are strongly connected with our choice of seed vectors. In fact, using $\varphi_\omega=b\varphi_0$ and $\psi_\omega=a^\dagger\psi_0$ (then $H\varphi_\omega=\omega \varphi_\omega$ and $H^\dagger\psi_\omega=\omega n\psi_\omega$) rather than $\varphi_0$ and $\psi_0$, it is clear, for instance, that $a^n\varphi_\omega=0$ only for $n\geq2$. So we see that the vectors $\varphi_{\omega:j,n}$ are different from the $\varphi_{0:j,n}$ (even if they are clearly related). In other words, it is not essential to start from the ground states of $H$ and $H^\dagger$ to use our ALOs.

Going back to $\varphi_{0:j,n}$ and $\psi_{0:j,n}$, we find
$$
H\varphi_{0:2,n}=\omega n\varphi_{0:2,n}, \qquad H^\dagger\psi_{0:1,n}=\omega n\psi_{0:1,n},
$$
while $\varphi_{0:1,n}=\psi_{0:2,n}=0$ for all $n\geq1$. These eigenvalue equations are in agreement with 
formula (\ref{47}) since we have $\epsilon_{0:2,n}=0+n\lambda_2=\omega n$ and $\epsilon_{0:1,n}=0+n\lambda_1=-\omega n$, so that $\overline{\epsilon_{0:1,-n}}=\omega n$. Formula (\ref{48}) becomes
$$
\langle\psi_{0:1,m},\varphi_{0:2,n}\rangle=n!\delta_{n,m},
$$
while the other scalar products are trivial. It is clear that, since $\lambda_1+\lambda_2=0$, $[H,Z_1Z_2] = [H,Z_2Z_1]=0$ and  $[H^\dagger,Z_1^\dagger Z_2^\dagger] = [H^\dagger,Z_2^\dagger Z_1^\dagger]=0$. Hence $(H,Z_1,Z_2)$ satisfy Condition $N_2$, and equations (\ref{412}) and (\ref{413}) imply that $\varphi_{0:2,n}$ is eigenstate of $Z_1Z_2=ab$ and of $Z_2Z_1=ba$, while $\psi_{0:1,n}$ is eigenstate of $Z_1^\dagger Z_2^\dagger=a^\dagger b^\dagger$ and of $Z_2^\dagger Z_1^\dagger=b^\dagger a^\dagger$. A direct computation produces, for instance,
$$
Z_1Z_2\varphi_{0:2,n}=(n+1)\varphi_{0:2,n}, \qquad Z_2Z_1\varphi_{0:2,n}=n\varphi_{0:2,n}.
$$
These two equalities show what stated right after (\ref{414}): $n+1=z_{0:2,n}(1,2)\neq z_{0:2,n}(2,1)=n$: the order of the $Z_j$ is important.

\subsubsection{Pseudo-bosons in $d=2$}\label{exapbsd2}

The previous example can be extended to a two-dimensional systems without particular difficulties, leaving untouched the algebraic settings, $\LD$. Of course, in this case, $\D=\Sc(\mathbb{R}^2)$, the set of test functions in two-dimensions, or some other useful dense subset of $\Lc^2(\mathbb{R}^2)$.

Let $a_j$ and $b_j$ be again $\D$-pseudo bosonic operators. They satisfy $[a_j,b_k]=\delta_{j,k}\1$, $j,k=1,2$. Let us put $N_j=b_ja_j$, and $H=\omega_1N_1+\omega_2N_2$, with $\omega_1$ and $\omega_2$ real and non zero. It is known, \cite{baginbagbook,bagrev}, that the eigenvalues of $H$ and $H^\dagger$ are the same, $E_{\bf n}=E_{n_1,n_2}=\omega_1n_1+\omega_2n_2$, while their eigenstates can be constructed extending what we have done in the previous, 1-d, example. For this, we start with two non zero vacua, $\varphi_{\bf0}=\varphi_{0,0}$ and $\psi_{\bf 0}= \psi_{0,0}$ in $\D$, satisfying $a_j\varphi_{\bf 0}=b_j^\dagger\psi_{\bf 0}=0$, $j=1,2$. Then we construct two families of vectors, all in $\D$, as follows:
$$
\varphi_{\bf n}=\varphi_{n_1,n_2}=\frac{1}{\sqrt{n_1!\,n_2!}}\,b_1^{n_1}b_2^{n_2}\varphi_{\bf0}, \qquad \psi_{\bf n}=\psi_{n_1,n_2}=\frac{1}{\sqrt{n_1!\,n_2!}}\,{a_1^\dagger}^{n_1}{a_2^\dagger}^{n_2}\psi_{\bf 0},
$$
$n_1,n_2\geq0$, and
$$
H\varphi_{\bf n}=E_{\bf n}\varphi_{\bf n}, \qquad H^\dagger\psi_{\bf n}=E_{\bf n}\psi_{\bf n}.
$$
In particular, $H\varphi_{\bf 0}=H^\dagger\psi_{\bf 0}=0$. These vectors will be the starting point of our construction. 

First of all, we notice that all the eigenvalues of $H$ and $H^\dagger$ have multiplicity one if $\frac{\omega_1}{\omega_2}\notin\mathbb{Q}$. As we have seen, this is a sufficient (but not necessary) condition to ensure that the vectors we will construct later are eigenstates of different operators, at least if Condition $N_0$ holds for some $N_0$. More explicitly, if we put
$$
Z_1=a_1, \quad Z_2=b_1,\quad Z_3=a_2, \quad z_4=b_2, \quad \mbox{ and }\quad \lambda_1=-\lambda_2=-\omega_1, \quad \lambda_3=-\lambda_4=-\omega_2,
$$
it follows that
$$
[H,Z_j]=\lambda_j Z_j,
$$
$j=1,2,3,4$. Now, the vectors in (\ref{45}) are the following:
$$
\varphi_{0:j,n}=Z_j^n\varphi_0=\left\{
\begin{array}{ll}
	a_1^n\varphi_{\bf 0}=0\hspace{2,2cm}\, \forall n\geq 1, \quad j=1\\
	b_1^n\varphi_{\bf 0}=\sqrt{n!}\,\,\varphi_{n,0}, \qquad \forall n\geq 0,\quad j=2\\
	a_2^n\varphi_{\bf 0}=0\hspace{2,2cm}\, \forall n\geq 1, \quad j=3\\
b_2^n\varphi_{\bf 0}=\sqrt{n!}\,\,\varphi_{0,n}, \qquad \forall n\geq 0,\quad j=4,\\
\end{array}
\right.
$$
while
$$
\psi_{0:j,n}={Z_j^\dagger}^n\psi_0=\left\{
\begin{array}{ll}
{a_1^\dagger}^n\psi_{\bf 0}=\sqrt{n!}\,\,\psi_{n,0}\hspace{1.7cm}\, \forall n\geq 1, \quad j=1\\
{b_1^\dagger}^n\psi_{\bf 0}=0, \hspace{3cm} \forall n\geq 0,\quad j=2\\
{a_2^\dagger}^n\psi_{\bf 0}=\sqrt{n!}\,\,\psi_{0,n}\hspace{1.7cm}\, \forall n\geq 1, \quad j=3\\
{b_2^\dagger}^n\psi_{\bf 0}=0, \hspace{3cm} \forall n\geq 0,\quad j=4,\\
\end{array}
\right.
$$
and we see that
$$
H\varphi_{0:2,n}=\omega_1n\varphi_{0:2,n}, \quad H\varphi_{0:4,n}=\omega_2n\varphi_{0:4,n},\quad H^\dagger \psi_{0:1,n}=\omega_1 n\psi_{0:1,n}\quad H^\dagger \psi_{0:3,n}=\omega_2 n\psi_{0:3,n},
$$
that
$$
\omega_{0:2,n}=0+n\lambda_2=n\omega_1, \qquad \omega_{0:4,n}=0+n\lambda_4=n\omega_2, 
$$ and that
$$ \overline{\omega_{0:1,-n}}=\overline{0-n\lambda_2}=n\omega_1, \qquad \overline{\omega_{0:3,-n}}=\overline{0-n\lambda_3}=n\omega_2,
$$
so that formula (\ref{46}) is recovered. As for the biorthogonality of the eigenvectors, we have
$$
\langle\varphi_{0:2,n},\psi_{0:3,m}\rangle=\langle\varphi_{0:4,m},\psi_{0:1,n}\rangle=0,
$$
for all $n,m\geq1$, while
$$
\langle\varphi_{0:2,n},\psi_{0:1,m}\rangle=\langle\varphi_{0:4,n},\psi_{0:3,m}\rangle=n!\,\delta_{n,m}, 
$$
in agreement with (\ref{48}). 

Now, since $\lambda_1+\lambda_2=\lambda_3+\lambda_4=0$, it follows that $Z_1Z_2=a_1b_1=N_1+\1$, $Z_2Z_1=N_1$, $Z_3Z_4=N_2+\1$ and $Z_4Z_3=N_2$, all commute with $H$. It is now clear that $\varphi_{0:j,n}$, $j=2,4$, are eigenstates of all these operators, while the $\psi_{0:j,n}$, $j=1,3$, are eigenstates of their adjoints. As in the previous example, we can explicitly check that, for instance, $n+1=z_{0:2,n}(1,2)\neq z_{0:2,n}(2,1)=n$. Similarly we have $n+1=z_{0:4,n}(3,4)\neq z_{0:4,n}(4,3)=n$. Once again we see that the order of the ALOs is important in the determination of the eigenvalues, while the eigenvectors are always the same.

\vspace{3mm}

It is now interesting to see how this example looks like in some concrete, and more {\em explicit}, situation. We first consider the following manifestly non self-adjoint Hamiltonian
$$
H=\frac{1}{2}(p_1^2+x_1^2)+\frac{1}{2}(p_2^2+x_2^2)+i\left[A(x_1+x_2)+B(p_1+p_2)\right], $$ where $A$ and $B$ are real constants,
while $x_j$ and $p_j$ are the  self-adjoint position and momentum operators, satisfying $[x_j,p_k]=i\delta_{j,k}\1$, \cite{miao,baglat}. Since $H$ is quadratic in its variables, it is reasonable to look for ladder operators which are linear in the $x_j$'s and $p_j$'s. Indeed, if we put $C=-B+iA$ and $D=B+iA$, it is possible to check that
$$
Z_1=\frac{1}{\sqrt{2}}(x_1+ip_1+C), \quad Z_2=\frac{1}{\sqrt{2}}(x_1-ip_1+D), \quad Z_3=\frac{1}{\sqrt{2}}(x_2+ip_2+C), \quad Z_4=\frac{1}{\sqrt{2}}(x_2-ip_2+D) 
$$
satisfy (\ref{41}) with $\lambda_1=\lambda_3=-1$ and  $\lambda_2=\lambda_4=1$. This is easy to understand in terms of pseudo-bosonic operators, since $H$ can be rewritten as in the beginning of Section \ref{exapbsd2}, a part for the vacuum energy: $
H=N_1+N_2+(A^2+B^2+1)\1, 
$
where, as usual, $N_j=b_ja_j$, $[a_j,b_k]=\delta_{j,k}\1$, $a_j=c_j+\frac{C}{\sqrt{2}}$, $b_j=c_j^\dagger+\frac{D}{\sqrt{2}}$, and $c_j=\frac{1}{\sqrt{2}}(x_j+ip_j)$ is the standard bosonic annihilation operator for the $j$-th mode. Of course, for this Hamiltonian, $\omega_1=\omega_2=1$, and not all the eigenvalues have multiplicity one. However, most of the general results deduced before can still be explicitly deduced, and even more. In particular, since $\lambda_1+\lambda_2=\lambda_3+\lambda_4=0$, we observe that $Z_1Z_2$, $Z_2Z_1$, $Z_3Z_4$ and $Z_4Z_3$ commute with $H$ (and their adjoints commute with $H^\dagger$). But since we also have $\lambda_1+\lambda_4=\lambda_3+\lambda_2=0$, also $Z_1Z_4=a_1b_2$ and $Z_2Z_3=b_1a_2$ commute with $H$. Notice that $Z_4Z_1=Z_1Z_4$ and $Z_2Z_3=Z_3Z_2$: in this case the order of these operators is not relevant. Notice also that, when written in terms of the original variables $x_j$ and $p_j$, $Z_1Z_4$ and $Z_2Z_3$ appear as non trivial quadratic operators. The fact that the $\varphi_{0:j,n}$, $j=2,4$, are eigenstates of all these operators $Z_jZ_k$ is easy to check, and will not be done here.

\vspace{2mm}

{\bf Remark:--} The appearance of more operators commuting with $H$ is a  consequence of the existence of degeneracy of the eigenvalues of $H$. Indeed, $Z_2Z_3$ does not commute with $H$ if $\frac{\omega_1}{\omega_2}\notin\mathbb{Q}$, and this is because $\lambda_2+\lambda_3=\omega_1-\omega_3\neq0$, in this case. In other words, the presence of degeneracies  enriches the model.

\vspace{3mm}

We end this section by briefly considering another manifestly non self-adjoint Hamiltonian, again connected to a two-dimensional harmonic oscillator, which can be rewritten in terms of 2-d pseudo-bosonic operators. We consider
$$ H=(p_1^2+x_1^2)+(p_2^2+x_2^2+2ix_2)+2\epsilon x_1x_2,$$ where $\epsilon\in]-1,1[$. This Hamiltonian was proposed and studied in \cite{ben1,miao1,baglat}, where it is shown that, putting
$$
\left\{
\begin{array}{ll}
a_1=\frac{1}{2\sqrt[4]{1+\epsilon\, \xi}}\left((ip_1+\sqrt{1+\epsilon\, \xi}\,x_1)+\xi(ip_2+\sqrt{1+\epsilon\, \xi}\,x_2)+
i\,\frac{\xi}{\sqrt{1+\epsilon\, \xi}}\right),\\
a_2=\frac{1}{2\sqrt[4]{1-\epsilon\, \xi}}\left((ip_1+\sqrt{1-\epsilon\, \xi}\,x_1)-\xi(ip_2+\sqrt{1-\epsilon\, \xi}\,x_2)-
i\,\frac{\xi}{\sqrt{1-\epsilon\, \xi}}\right),\\
b_1=\frac{1}{2\sqrt[4]{1+\epsilon\, \xi}}\left((-ip_1+\sqrt{1+\epsilon\, \xi}\,x_1)+\xi(-ip_2+\sqrt{1+\epsilon\, \xi}\,x_2)+
i\,\frac{\xi}{\sqrt{1+\epsilon\, \xi}}\right),\\
b_2=\frac{1}{2\sqrt[4]{1-\epsilon\, \xi}}\left((-ip_1+\sqrt{1-\epsilon\, \xi}\,x_1)-\xi(-ip_2+\sqrt{1-\epsilon\, \xi}\,x_2)-
i\,\frac{\xi}{\sqrt{1-\epsilon\, \xi}}\right),
\end{array}
\right.
$$
we can rewrite
$$ H=H_1+H_2+\frac{1}{1-\epsilon^2}\,\1,\qquad
H_1=\sqrt{1+\epsilon\, \xi}\,(2N_1+\1), \quad H_2=\sqrt{1-\epsilon\, \xi}\,(2N_2+\1). $$ 
Here $\xi$ can be either $+1$ or $-1$, and $N_j=b_ja_j$, as in the previous example. These are $\D$-PBs, \cite{baglat}, so that, in particular, $[a_j,b_k]=\delta_{j,k}\1$. The eigenvalues of $H$ are the following:
$$
E_{\bf n}=E_{n_1,n_2}=\sqrt{1+\epsilon\, \xi}\,(2n_1+\1)+\sqrt{1-\epsilon\, \xi}\,(2n_2+\1)+\frac{1}{1-\epsilon^2},
$$
and their multiplicity is one if $\epsilon$ is chosen in such a way $\sqrt{\frac{1+\epsilon}{1-\epsilon}}\notin\mathbb{Q}$. In this case the operators $Z_j$ are the following:
$$
Z_1=a_1, \qquad Z_2=b_1,\qquad Z_3=a_2, \qquad z_4=b_2, 
$$
and the corresponding $\lambda_j$ are
$$
\lambda_1=-\lambda_2=-2\sqrt{1+\epsilon\, \xi}, \qquad \lambda_3=-\lambda_4=-2\sqrt{1-\epsilon\, \xi}.
$$
It is clear that now, while $Z_1Z_2$, $Z_2Z_1$, $Z_3Z_4$ and $Z_4Z_3$ still commute with $H$, $Z_1Z_3$ and $Z_2Z_4$ do not. This is a consequence of the lack of degeneracy in $\sigma_p(H)$.

\section{Conclusions}\label{sectconcl}

We have analyzed some general aspects of  ALOs, in presence of Hamiltonians which can be self-adjoint or not. We have  shown that ALOs can be defined for generic factorizable Hamiltonians, and for Hamiltonians costructed in terms of generalized Heisenberg algebra. In all these cases, we have seen that some initial requirements on the Hamiltonian $H_0=H_0^\dagger$ gives us the possibility to construct, in a rather specific way, a family of vectors which are, when non zero, eigenvectors of $H_0$, and to deduce the related eigenvalues. This procedure, which was known and already used by some authors in some specific case, has been generalized to let new cases to fit in. Some applications have been discussed to clarify the general results. What is still missing, and in our opinion deserves a deeper investigation, is the possibility of using the settings proposed in Section \ref{sect4} to deal with the extended generalized Heisenberg algebra considered in \cite{bcg}. This is work in progress.

\section*{Acknowledgements}

The author acknowledges partial support from Palermo University and from G.N.F.M. of the INdAM.

\renewcommand{\theequation}{A.\arabic{equation}}

\section*{Appendix: $O^*$-algebras}\label{appendix}

Let us briefly review how  $\Lc^\dagger(\D)$ can be introduced, and why it is so relevant for us.  We refer to \cite{aitbook,trrev,bagrev2007} for many results on $*$-algebra, { quasi $*$-algebras}, and $O^*$-algebras. In particular, we have:

\begin{defn}\label{o*}Let $\mathcal{H}$ be a separable Hilbert space and $N_0$ an
	unbounded, densely defined, self-adjoint operator. Let $D(N_0^k)$ be
	the domain of the operator $N_0^k$, $k \ge 0$, and $\mathcal{D}$ the domain of
	all the powers of $N_0$, that is,  $$ \mathcal{D} = D^\infty(N_0) = \bigcap_{k\geq 0}
	D(N_0^k). $$ This set is dense in $\mathcal{H}$. We call
	$\mathcal{L}^\dagger(\mathcal{D})$ the $*$-algebra of all \textit{  closable operators}
	defined on $\mathcal{D}$ which, together with their adjoints, map $\mathcal{D}$ into
	itself. Here the adjoint of $X\in\mathcal{L}^\dagger(\mathcal{D})$ is
	$X^\dagger=X^*_{| \mathcal{D}}$. $\mathcal{L}^\dagger(\mathcal{D})$ is called  an $O^*$-algebra.
\end{defn}

In $\mathcal{D}$ the topology is defined by the following $N_0$-depending
seminorms: $$\phi \in \mathcal{D} \rightarrow \|\phi\|_n\equiv \|N_0^n\phi\|,$$
where $n \ge 0$, and  the topology $\tau_0$ in $\mathcal{L}^\dagger(\mathcal{D})$ is introduced by the seminorms
$$ X\in \mathcal{L}^\dagger(\mathcal{D}) \rightarrow \|X\|^{f,k} \equiv
\max\left\{\|f(N_0)XN_0^k\|,\|N_0^kXf(N_0)\|\right\},$$ where
$k \ge 0$ and   $f \in \mathcal{C}$, the set of all the positive,
bounded and continuous functions  on $\mathbb{R}_+$, which are
decreasing faster than any inverse power of $x$:
$\mathcal{L}^\dagger(\mathcal{D})[\tau_0]$ is a {   complete *-algebra}.

The relevant aspect of $\LD$ is that, \cite{aitbook,bagrev2007,trrev}, if $x,y\in \mathcal{L}^\dagger(\mathcal{D})$, we can multiply them and the results, $xy$ and $yx$, both belong to $\mathcal{L}^\dagger(\mathcal{D})$, as well as their difference, the commutator $[x,y]$. Also, powers of $x$ and $y$ all belong to $\Lc^\dagger(\D)$, which is therefore a good framework to work with, also in presence of unbounded operators. In fact, if $N_0=a^\dagger a$, where $[a,a^\dagger]=\1_b$ as in Section \ref{sect2}, we can prove that $a, a^\dagger\in\Lc^\dagger(\D)$. Hence $N_0\in  \Lc^\dagger(\D)$ as well. This is also true for pseudo-bosonic operators, \cite{bagrev,bagrus2018}, at least if $\D=\Sc(\mathbb{R})$.

\end{document}